\documentclass[a4paper,11pt,reqno]{amsart}

\usepackage[margin=35mm]{geometry}
\usepackage[T1]{fontenc}
\usepackage[english]{babel}
\usepackage[utf8]{inputenc}
\usepackage{amsmath,amssymb,amsthm,amsfonts,amstext,amsopn,amsxtra,dsfont,mathrsfs,esint,enumitem,bookmark,mathtools}
\usepackage{bbm}
\usepackage[capitalise]{cleveref}
\hypersetup{linktocpage=true}
\usepackage{microtype}
\usepackage[headings]{fullpage}\usepackage{microtype}
\usepackage[headings]{fullpage}

\newcommand\numberthis{\addtocounter{equation}{1}\tag{\theequation}}

\theoremstyle{plain}
\newtheorem{theorem}{Theorem}
\newtheorem{lemma}[theorem]{Lemma}
\newtheorem{proposition}[theorem]{Proposition}

\theoremstyle{definition}

\newtheorem{assumption}[theorem]{Assumption}

\newcommand\xqed[1]{%
	\leavevmode\unskip\penalty9999 \hbox{}\nobreak\hfill\quad\hbox{#1}%
}
\newcommand\remarkend{\xqed{$\triangle$}}

\makeatletter
\def\@endtheorem{\remarkend\endtrivlist\@endpefalse }
\makeatother
\theoremstyle{remark}

\makeatletter
\def\@endtheorem{\endtrivlist\@endpefalse }
\makeatother
\crefname{theorem}{Theorem}{Theorems}
\crefname{lemma}{Lemma}{Lemmas}
\crefname{proposition}{Proposition}{Propositions}
\crefname{corollary}{Corollary}{Corollaries}
\crefname{definition}{Definition}{Definitions}
\crefname{assumption}{Assumption}{Assumptions}
\crefname{remark}{Remark}{Remarks}

\crefname{subsection}{subsection}{subsections}
\crefname{subsubsection}{subsection}{subsections}

\renewcommand{\d}[1]{\ensuremath{\operatorname{d}\!{#1}}}

\renewcommand{\d}[1]{\ensuremath{\operatorname{d}\!{#1}}}

\newcommand{\Tr}{\operatorname{Tr}}

\newcommand{\curl}{\operatorname{curl}}
\newcommand{\s}{\operatorname{s}}
\newcommand{\nls}{{\operatorname{nls}}}
\newcommand{\op}{\operatorname{op}}
\newcommand{\imax}{i_{\mathrm{max}}}

\renewcommand{\P}{\mathbb{P}}
\newcommand{\Q}{\mathbb{Q}}
\newcommand{\R}{\mathbb{R}}

\newcommand{\mfH}{\mathfrak{H}}
\newcommand{\mfS}{\mathfrak{S}}
\newcommand{\mfX}{\mathfrak{X}}

\newcommand{\cE}{\mathcal{E}}

\newcommand\mydots{\ifmmode\mathellipsis\else.\kern-0.08em.\kern-0.08em.\fi}

\allowdisplaybreaks

\DeclareUnicodeCharacter{2217}{*}

\usepackage{csquotes}
\usepackage[backend=bibtex,sorting=nty,citestyle=numeric-comp,bibstyle=ieee,maxbibnames=99,doi=false,isbn=false,url=false,eprint=false,dashed=true]{biblatex}
\AtEveryBibitem{\clearfield{month}}
\AtEveryBibitem{\clearfield{day}}
\usepackage{graphicx}
\usepackage{subcaption}
\usepackage[table,xcdraw]{xcolor}

\usepackage{tikz}
\usetikzlibrary{arrows}
\usepackage{tkz-tab}

\usepackage{standalone}

\usepackage{hyperref}
\hypersetup{bookmarksdepth=3}
\DeclareRobustCommand{\SkipTocEntry}[9]{}

\title[Derivation of Hartree theory in 2D]{Derivation of Hartree theory for two-dimensional attractive Bose gases in almost Gross--Pitaevskii regime}
\author[L. Junge]{Lukas Junge}
\address{Department of Mathematics, Copenhagen university, Lyngbyvej 2, 2100 Copenhagen, Denmark}
\email{lj@math.ku.\d{}k}
\author[F. L. A. Visconti]{François L. A. Visconti}
\address{Department of Mathematics, LMU Munich, Theresienstrasse 39, 80333 Munich, Germany}
\email{visconti@math.lmu.de}
\date{\today}

\begin{document}
	\begin{abstract}
		We study the ground state energy of trapped two-dimensional Bose gases with mean-field type interactions that can be attractive. We prove the stability of second kind of the many-body system and the convergence of the ground state energy per particle to that of a non-linear Schrödinger (NLS) energy functional. Notably, we can take any polynomial scaling of the interaction, and even exponential scalings arbitrarily close to the Gross--Pitaevskii regime, which is a drastic improvement on the best-known result for systems with attractive interactions. As a consequence of the stability of second kind we also obtain Bose--Einstein condensation for the many-body ground states for a much improved range of the diluteness parameter.
	\end{abstract}
		
	\maketitle
		
	\tableofcontents
	
	\section{Introduction}
	
	We consider $N$ two-dimensional bosons interacting via a pair potential $w_N$ and trapped by an external potential $V$. The system is described by the N-body Hamiltonian
	\begin{equation}
		\label{eq:hamiltonian}
		H_N = \sum_{i=1}^N\left(-\Delta_{x_i} + V(x_i)\right) + \frac{1}{N-1}\sum_{1\leq i<j\leq N}w_N(x_i - x_j)
	\end{equation}
	acting on
	\begin{equation*}
		\mfH^N = \bigotimes_\textmd{sym}^N\mfH,
	\end{equation*}
	the symmetric tensor product of $N$ copies of the one-body Hilbert space $\mfH \coloneqq L^2(\mathbb{R}^2)$. The interaction potential $w_N$ and the trapping potential $V$ satisfy the following assumptions.
	\begin{assumption}
		\label{ass:potentials}
		The two-body interaction is either of the form
		\begin{equation}
			\label{eq:interaction_potential_beta_scaling}
			w_N = N^{2\beta}w(N^\beta\cdot),
		\end{equation}
		for some fixed $\beta > 0$, or of the form
		\begin{equation}
			\label{eq:interaction_potential_kappa_scaling}
			w_N = e^{2N^\kappa}w(e^{N^\kappa}\cdot),
		\end{equation}
		for some $0 < \kappa < 1$. The function $w : \R^2 \rightarrow \R$ is fixed and satisfies
		\begin{equation*}
			w \in L^1(\R^2)\cap L^{1 + \eta}(\R^2) \quad \textmd{and} \quad w(x) = w(-x),
		\end{equation*}
		for some $\eta > 0$. The external trapping potential $V:\mathbb{R}^2 \rightarrow \mathbb{R}^+$ satisfies
		\begin{equation}
			\label{eq:trapping_potential_assumption}
			V(x) \geq C^{-1}|x|^s - c
		\end{equation}
		for some constants $s > 0$ and $c,C > 0$.
	\end{assumption}
	Under these assumptions the Hamiltonian $H_N$ is bounded from below with the core domain $\mfH^N \cap C_{\textmd{c}}^\infty(\mathbb{R}^{2N})$ and can thus be extended to a self-adjoint operator by Friedrichs' method.
	
	The system we are considering is a mean-field system and is expected to exhibit Bose--Einstein condensation, meaning that almost all particles would live in the same quantum state. In terms of wavefunctions this roughly translates to
	\begin{equation}
		\label{eq:wavefunction_product_state_ansatz}
		\Psi(x_1,\dots,x_N) \approx u^{\otimes N}(x_1,\dots,x_N) = u(x_1)\cdots u(x_N).
	\end{equation}
	Taking the trial state wavefunction $u^{\otimes N}$, it is therefore natural to look at the Hartree energy functional
	\begin{equation}
		\label{eq:hartree_functional}
		\cE^\textmd{H}_N[u] = \dfrac{\left\langle u^{\otimes N}, H_Nu^{\otimes N}\right\rangle}{N} = \langle u,hu\rangle + \dfrac{1}{2}\int_{\mathbb{R}^2}(w_N * |u|^2)|u|^2,
	\end{equation}
	where we defined the 1-body Hamiltonian
	\begin{equation*}
		h = -\Delta + V.
	\end{equation*}
	Since
	\begin{equation*}
		w_N \rightharpoonup a\delta_0 \quad \textmd{with} \quad a \coloneqq \int_{\mathbb{R}^2} w
	\end{equation*}
	in the limit $N \rightarrow \infty,$ the Hartree functional \eqref{eq:hartree_functional} \textit{formally} converges to the non-linear Schrödinger (NLS) energy functional
	\begin{equation}
		\label{eq:nls_energy_functional}
		\cE^{\textmd{nls}}[u] = \left\langle u,hu\right\rangle + \dfrac{a}{2}\int_{\mathbb{R}^2}|u|^4.
	\end{equation}
	The speed at which the convergence of the Hartree functional \eqref{eq:hartree_functional} to the NLS functional \eqref{eq:nls_energy_functional} occurs is measured by the parameter $\beta$ in the polynomial case and $\kappa$ in the exponential case.
	
	For the repulsive case $w \geq 0$, it is known that the limiting energy functional \eqref{eq:nls_energy_functional} admits a subtle correction when the potential scales exponentially with $N$. Namely, when taking $\kappa = 1$ in \eqref{eq:interaction_potential_kappa_scaling} and removing the mean-field factor $(N - 1)^{-1}$ in \eqref{eq:hamiltonian}, we obtain the so-called \textit{Gross--Pitaevskii regime}, for which the convergence to point-wise like interactions still occurs, but where the limiting functional is now \eqref{eq:nls_energy_functional} with $a \approx 8\pi/\log(N/\mathfrak{a}^2)$, where $\mathfrak{a}$ is the scattering length of $w$ \cite{Lieb2005MathematicsBG,Lieb2006DerivationGP,Lieb2001RigorousDGP}. This is because the ground state of \eqref{eq:hamiltonian} includes a non-trivial correction to the ansatz \eqref{eq:wavefunction_product_state_ansatz}, in the form of a short-range correlation structure.
	
	When $w$ is not purely repulsive (e.g. $w \leq 0$), the problem is more difficult because we have to deal with the issue of the system's stability. More precisely, the energy functional $\cE^{\textmd{nls}}$ is bounded from below under the constraint $\|u\|_2 = 1$ if and only if
	\begin{equation}
		\label{eq:condition_nls_bdd_below}
		a \geq -a ^*,
	\end{equation}
	where $a^*$ is the optimal constant of the Gagliardo--Nirenberg inequality
	\begin{equation}
		\label{eq:gagliardo_nirenberg_inequality}
		\left(\int_{\mathbb{R}^2}\vert \nabla u\vert^2\right) \left(\int_{\mathbb{R}^2}\vert u\vert^2\right)\geq \frac{a^*}{2}\int_{\mathbb{R}^2} \vert u\vert^4, \quad \forall u\in H^1\left(\mathbb{R}^2\right).
	\end{equation}
	See \cite{Guo2014MassConcentration,Maeda2010SymGSE,Weinstein1983nls,Zhang2000StabilityAttractiveBEC} for references. By the variational construction of \eqref{eq:nls_energy_functional}, the condition \eqref{eq:condition_nls_bdd_below} is thus necessary for stability of second kind \cite{Lieb2010Stability} of the many-body system:
	\begin{equation}
		\label{eq:stability_second_kind_intro}
		H_N \geq -CN.
	\end{equation}
	Sufficient conditions for stability of second kind are however more restrictive, and we work with
	\begin{equation}
		\label{eq:interaction_potential_negative_assumption}
		\int_{\R^2}w^- < a^*,
	\end{equation}
	where $w_-$ denotes the negative part of $w$, given by $-\min(w,0)$. When the potential $w$ is nonpositive, the condition \eqref{eq:interaction_potential_negative_assumption} is just the strict version of \eqref{eq:condition_nls_bdd_below}. We refer to \cite{Lewin2014TheMA} for a refinement called Hartree stability.
	
	The goal of the present paper is to show that \eqref{eq:interaction_potential_negative_assumption} is a \textit{sufficient} condition to ensure the stability of second kind of the many-body system \eqref{eq:stability_second_kind_intro}. In the polynomial case \eqref{eq:interaction_potential_beta_scaling} and under the condition \eqref{eq:interaction_potential_negative_assumption}, this has been proved in \cite{Lewin2014TheMA,Lewin2017Note2D} for $\beta>0$ sufficiently small, and later in \cite{Nam2020ImprovedStability} for $0< \beta< 1$. In the present work, we extend this stability result to all $\beta \in (0,\infty)$ and to exponential scalings \eqref{eq:interaction_potential_kappa_scaling} for all $\kappa\in(0,1)$.
	
	Alongside proving stability of second kind, we also prove the convergence of the many-body ground state energy to that of the NLS energy functional \eqref{eq:nls_energy_functional} for any $\beta \in (0,\infty)$ (or any $\kappa \in(0,1)$). Consequently, we also obtain the convergence of the many-body ground states to those of the NLS functional \eqref{eq:nls_energy_functional}. Though both the \textit{defocusing} ($a \geq 0$) and \textit{focusing} ($a \leq 0$) cases are covered, the main novelty of the paper lies in the latter case.
	\\
	
	\noindent
	\textbf{Acknowledgments.}
	We thank Phan Thành Nam for his continuous support and his precious feedback. We also thank the anonymous referee. L. J. was partially supported by the European Union. Views and opinions expressed are however those of the authors only and do not necessarily reflect those of the European Union or the European Research Council. Neither the European Union nor the granting authority can be held responsible for them. L. J. was partially supported by the Villum Centre of Excellence for the Mathematics of Quantum Theory (QMATH) with Grant No.10059. L. J. was supported by the grant 0135-00166B from Independent Research Fund Denmark. F. L. A. V. acknowledges partial support by the Deutsche Forschungsgemeinschaft (DFG, German Research Foundation) through the TRR 352 Project ID. 470903074 and by the European Research
	Council through the ERC CoG RAMBAS Project Nr. 101044249.
	
	\section{Main result}
	
	\subsection{Notations}
	
	Let $S_N$ be the group of permutations of $\left\{1,\dots,N\right\}$. For $\psi_1 \in \mfH^{N_1}$ and $\psi_2 \in \mfH^{N_2}$, we define the symmetric tensor product $\psi_1\otimes_\textmd{s}\psi_2\in\mfH^{N_1 + N_2}$ by
	\begin{multline*}
		\psi_1\otimes_\textmd{s}\psi_2(x_1,\dots,x_{N_1 + N_2})\\
		= \dfrac{1}{\sqrt{N_1!N_2!(N_1 + N_2)!}}\sum_{\sigma\in S_{N_1 + N_2}}\psi_1(x_{\sigma(1)},\dots,x_{\sigma(N_1)})\psi_2(x_{\sigma(N_1 + 1)},\dots,x_{\sigma(N_1 + N_2)}).
	\end{multline*}
	We denote the $i$-fold tensor product of a vector $f\in\mfH$ by $f^{\otimes i}\in \mfH^i$, and the $i$-fold tensor product of an operator $A:\mfH\rightarrow\mfH$ by $A^{\otimes i}$. Let $\mfS^1(\mfX)$ be the space of all trace-class operators on a given Hilbert space $\mfX$ \cite{Simon1979TraceIdeals}. Define the set $\mathcal{S}(\mfX)$ of all states on $\mfX$ by
	\begin{equation*}
		\mathcal{S}(\mfX) = \left\{\Gamma \in \mfS^1(\mfX): \Gamma = \Gamma^* \geq 0, \Tr_\mfX\Gamma = 1\right\}
	\end{equation*}
	The $k$-particle reduced density matrix $\Gamma^{(k)}$ of a given state $\Gamma\in \mathcal{S}(\mfH^N)$ is obtained by taking the partial trace over all but the first $k$ variables:
	\begin{equation*}
		\Gamma^{(k)} = \Tr_{k + 1\rightarrow N}\Gamma.
	\end{equation*}
	Moreover, we define the $k$-particle reduced density matrix $\gamma_{\Psi}^{(k)}$ of a normalised wavefunction $\Psi\in\mfH^N$ by
	\begin{equation*}
		\gamma_{\Psi}^{(k)} = \Tr_{k + 1\rightarrow N}\vert\Psi\rangle\langle\Psi\vert.
	\end{equation*}
	
	\subsection{Statement of the main result}
	
	Define the many-body ground state energy per particle
	\begin{equation*}
		e_N = \dfrac{1}{N}\inf\left\{\left\langle\Psi,H_N\Psi\right\rangle: \Psi\in\mfH^{N}, \left\Vert\Psi\right\Vert = 1\right\},
	\end{equation*}
	and the ground state energy of the NLS functional \eqref{eq:nls_energy_functional}
	\begin{equation*}
		e^\nls = \inf\left\{\cE^{\textmd{nls}}[u]: u\in\mfH, \Vert u\Vert = 1\right\}.
	\end{equation*}
	We prove the convergence of $e_N$ to $e^\nls$. The convergence of the ground states of $H_N$ to those of $\cE^{\textmd{nls}}$ follows directly thanks to arguments from \cite{Lewin2014TheMA}.
	\begin{theorem}[Convergence to NLS theory]
		\label{th:convergence_nls_theory}
		Let $w_N$ and $V$ satisfy Assumption~\ref{ass:potentials} and the condition \eqref{eq:interaction_potential_negative_assumption}. Then,
		\begin{equation*}
			\boxed{\lim_{N\rightarrow\infty}e_N = e_{\nls} > - \infty}
		\end{equation*}
		and 
		\begin{equation*}
			\boxed{H_N \geq -CN,}
		\end{equation*}
		for some constant $C > 0$ that depends only on $w$ and $V$. Moreover, for a sequence $\{\Psi_N\}_N$ of ground states of $H_N$, there exists a Borel probability measure $\mu$ supported on the minimisers of $\cE^{\textmd{nls}}$ such that, along a subsequence,
		\begin{equation*}
			\lim_{N\rightarrow\infty}\Tr\left\vert\gamma_{\Psi_N}^{(k)} - \int\vert u^{\otimes k}\rangle\langle u^{\otimes k}\vert\d{}\mu(u)\right\vert = 0,\quad \forall k\in\mathbb{N}.
		\end{equation*}
		If $\cE^{\textmd{nls}}$ has a unique minimiser $u_0$ (up to a phase), then for the whole sequence
		\begin{equation*}
			\lim_{N\rightarrow\infty}\Tr\left\vert\gamma_{\Psi_N}^{(k)} - \vert u_0^{\otimes k}\rangle\langle u_0^{\otimes k}\vert\right\vert = 0,\quad \forall k\in\mathbb{N}.
		\end{equation*}
	\end{theorem}
	
	\noindent
	\textit{Remarks.}
	\begin{enumerate}
		\item Thanks to the diamagnetic inequality \cite[Theorem 7.21]{Lieb2001Analysis}, we can easily add an external magnetic field $\mathbf{B} = \curl \mathbf{A}$ with $\mathbf{A}\in L^2_{\textmd{loc}}(\R^2)$ to the system, meaning that $-\Delta_{x_j}$ can be replaced by $(\mathbf{i}\nabla_{x_j} + \mathbf{A}(x_j))^2$ in \eqref{eq:hamiltonian}. Moreover, the same proof can also be used when working on the unit torus or in a finite box with Dirichlet boundary conditions instead of a trapping potential $V$. Heuristically, one can think of this as $s = \infty$ in \eqref{eq:trapping_potential_assumption}.
		\item When considering the dynamics of a two-dimensional Bose gas with attractive interactions, that is when trying to prove that Bose--Einstein condensation is preserved by the Hamiltonian flow $e^{\mathbf{i}tH_N}$, one is also confronted with issues of stability of second kind. The best-known result in this direction is the range $0 < \beta < 1$ and follows from the method of \cite{Jeblick2018DerivationOT} and the stability of second kind proven in \cite{Nam2020ImprovedStability}. It would be interesting to know whether Theorem~\ref{th:convergence_nls_theory} allows for an improvement of the previous result. More restrictive ranges had previously been obtained in \cite{Nam2019NormA} (without the use of \eqref{eq:stability_second_kind_intro}) and in \cite{Chen2017RigorousD2D} (for $s = 2$). See \cite{Chong2021DynamicsLBS,Pickl2010DerivationTD} for related results in 3D. In the defocusing case $w \geq 0$, we refer to \cite{Jeblick2019DerivationTimeGP,Kirkpatrick2011Derivation2DNLS} for results in 2D and \cite{Bossman2020Derivation2DGP,Chen2017RigorousD2D} for the effectively 2D dynamics of strongly confined 3D systems. We mention as well that the instability regime, that is $a < -a^*$ (see \eqref{eq:gagliardo_nirenberg_inequality}), poses natural and interesting questions for the dynamics, namely whether Bose--Einstein condensation persists until the blow-up of the Hartree solution. See \cite{Bossmann2023FocusingD2D} for new results in this direction.
		\item Stability of second kind is an issue as well for three-dimensional Bose gases with interaction potentials
		\begin{equation*}
			N^{3\beta}w(N^\beta\cdot)
		\end{equation*}
		having an attractive part in the dilute regime $\beta > 1/3$. For stability of second kind to hold one must make the additional assumption that the potential is classically stable \cite{Lewin2014TheMA,Triay2018DerivationDDGP}. Regrettably, a straightforward adaptation of the proof of Theorem~\ref{th:convergence_nls_theory} only works up to $\beta < 1/3$, thus failing to capture the dilute regime. The best-known results in the dilute regime are the ranges $1/3 < \beta < 1/3 + s/(45 + 42s)$ \cite{Triay2018DerivationDDGP}, where $s$ is the exponent in \eqref{eq:trapping_potential_assumption}, and $1/3 < \beta < 9/26$ \cite{Nam2020ImprovedStability}.
	\end{enumerate}
	
	\subsection{Strategy of the proof}
	
	We wish to compare the many-body ground state energy per particle $e_N$ to that of the Hartree functional \eqref{eq:hartree_functional}
	\begin{equation*}
		e_N^{\textmd{H}} = \inf\left\{\cE_N^{\textmd{H}}[u]:u\in\mfH,\Vert u\Vert = 1\right\},
	\end{equation*}
	and then use the convergence
	\begin{equation*}
		\label{eq:convergence_hartree_to_nls}
		\lim_{N\rightarrow \infty}e_N^{\textmd{H}} = e^\nls
	\end{equation*}
	(see \cite[Lemma 7]{Lewin2017Note2D}). The upper bound $e_N \leq e_N^{\textmd{H}}$ can be obtained immediately using the trial state $u^{\otimes N}$, and the difficult part is to prove a matching lower bound. Note that, for practical reasons, we shall not compare $e_N$ to $e_N^{\textmd{H}}$, but to the ground state $\tilde{e}_N^{\textmd{H}}$ of a slightly modified Hartree functional, which also converges to $e_\textmd{nls}$.
	
	To prove the lower bound, we shall use a localisation technique in momentum space in order to reduce the infinite dimensional problem to multiple finite dimensional ones (similarly to \cite{Lewin2014TheMA,Lewin2017Note2D,Nam2020ImprovedStability}). Then, we shall apply the following quantitative version of the quantum de Finetti theorem \cite{Brandao2017deFinetti,Li2015QF}.
	\begin{theorem}\label{th:de_finetti}
		Given a Hilbert space $\mfX$ of dimension $d$ and a symmetric state $\Gamma_K\in S(\mfX^K)$, there exists a probability measure $\mu$ on $S(\mfX)$ such that
		\begin{equation*}
			\left\vert\Tr \left[A\otimes B\left( \Gamma_K^{(2)} - \int_{S(\mfX)}\gamma^{\otimes 2} d\mu(\gamma)\right)\right]\right| \leq C\sqrt{\frac{\log d}{K}}\Vert A\Vert_{\op} \Vert B\Vert_{\op}
		\end{equation*}
		for all self-adjoint operators $A$ and $B$ on $\mfX$, and for some universal constant $C > 0$.
	\end{theorem}
	
	\begin{proof}
		In \cite{Brandao2017deFinetti}, the statement was proven with $A,B$ replaced by quantum measurement. See \cite[Proposition 3.3]{Rougerie2020NLSagain} and \cite[Lemma 3.3]{Nam2020ImprovedStability} for an adaptation to self-adjoint operators.
	\end{proof}
	
	We shall apply Theorem~\ref{th:de_finetti} on energy subspaces of the one-body Schrödinger operator $h$ acting on $\mfH$ defined by the spectral projections
	\begin{equation}
		\label{eq:spectral_projections_strategy_proof}
		P_1 = \mathds{1}_{\left\{h < N^{2\varepsilon}\right\}} \quad \textmd{and} \quad P_i = \mathds{1}_{\left\{N^{2(i-1)\varepsilon} \leq h < N^{2i\varepsilon}\right\}}, \; 2 \leq  i \leq M,
	\end{equation}
	for some $\varepsilon,M>0$. Note that thanks to Assumption \eqref{eq:trapping_potential_assumption} we have the Cwikel--Lieb--Rosenblum (CLR) type estimate
	\begin{equation}
		\label{eq:clr_estimate}
		\dim\left(\mathds{1}_{\left\{h < N^{2i\varepsilon}\right\}}\mfH\right) \leq CN^{(2 + 4/s)i\varepsilon},
	\end{equation}
	for all $i\in\{1,\dots,M\}$ (see \cite[Lemma 3.3]{Lewin2014TheMA} and references therein). When working on the unit torus or on a box with Dirichlet boundary conditions, the estimate \eqref{eq:clr_estimate} should be replaced by the usual Weyl asymptotic.
	
	Before applying Theorem~\ref{th:de_finetti} we shall write the energy of a ground state $\Psi$ of $H_N$ as
	\begin{equation*}
		\langle\Psi,H_N\Psi\rangle = N\Tr\big(H_{2,N}\Gamma^{(2)}\big),
	\end{equation*}
	using the two-particle reduced density matrix $\Gamma^{(2)}$ of $\Psi$ and the two-body Hamiltonian $H_{2,N}$ defined as
	\begin{equation*}
		H_{2,N} = (h_1 + h_2)/2 + w_N(x_1 - x_2)/2.
	\end{equation*}
	Decomposing the identity according to the spectral projections defined in \eqref{eq:spectral_projections_strategy_proof} we will then roughly get
	\begin{equation*}
		\Tr\big(H_{2,N}\Gamma^{(2)}\big) \gtrsim \sum_{\substack{1\leq i_1,i_2 \leq M\\ 1\leq i_1',i_2'\leq M}}\Tr\big(P_{i_1}\otimes P_{i_2}H_{2,N}P_{i_1'}\otimes P_{i_2'}\Gamma^{(2)}\big),
	\end{equation*}
	where the Sobolev inequality shall be used to neglect the terms containing the last spectral projection $P_{M + 1} \coloneqq \mathds{1}_{\left\{N^{2M\varepsilon} \leq h\right\}}$.
	
	After that, we will decompose the many-body state $\Psi$ according to the occupancy of its energy levels (defined by the spectral projections \ref{eq:spectral_projections_strategy_proof}). Namely, we decompose $\Psi$ as
	\begin{equation*}
		\Psi = \sum_{\underline{J}}\Psi_{\underline{J}} \quad \textmd{with } \Psi_{\underline{J}} = P_1^{\otimes j_1}\otimes_\textmd{s}\dots\otimes_\textmd{s}P_{M + 1}^{\otimes j_{M + 1}}\Psi,
	\end{equation*}
	where the sum is taken over all multi-indices $\underline{J} = (j_1,\dots,j_{M + 1})$ satisfying $|\underline{J}| = N$. Defining
	\begin{equation*}
		\Gamma_{\underline{J},\underline{J}'} = \vert\Psi_{\underline{J}'}\rangle\langle\Psi_{\underline{J}}\vert,
	\end{equation*}
	this decomposition gives
	\begin{equation*}
		\Tr\big(H_{2,N}\Gamma^{(2)}\big) \gtrsim \sum_{\underline{J},\underline{J}'}\sum_{\substack{1\leq i_1,i_2 \leq M\\ 1\leq i_1',i_2'\leq M}}\Tr\big(P_{i_1}\otimes P_{i_2}H_{2,N}P_{i_1'}\otimes P_{i_2'}\Gamma_{\underline{J},\underline{J}'}^{(2)}\big).
	\end{equation*}
	Then, we shall define the index
	\begin{equation*}
		i_{\textmd{max}}(\underline{J}) = \max\left\{i\in\left\{1,\dots,M\right\}:j_i \geq N^{1 - \delta\varepsilon}\right\},
	\end{equation*}
	for some $\delta > 0$, and make an important distinction between the terms that satisfy $i_{\textmd{max}}(\underline{J}) = i_{\textmd{max}}(\underline{J}')$ and for which $i_1,i_2,i_1'$ and $i_2'$ are all less than $i_{\textmd{max}}(\underline{J})$, and the other terms. In the latter case, we shall use Proposition~\ref{prop:plane_wave_estimate} and Lemma~\ref{lemma:state_two_projections_full_trace} (see below) to show that their potential energy is much smaller than the overall kinetic energy of the system and that their contribution to the many-body energy is therefore negligible. Doing so, we will roughly be left with
	\begin{equation}
		\label{eq:proof_strategy_lower_bound}
		\Tr\big(H_{2,N}\Gamma^{(2)}\big) \gtrsim \sum_{\substack{\underline{J},\underline{J}'\\ i_{\textmd{max}}(\underline{J}) = i_{\textmd{max}}(\underline{J}')}}\sum_{\substack{i_1,i_2 = 1\\ i_1',i_2' = 1}}^{i_{\textmd{max}}(\underline{J})}\Tr\big(P_{i_1}\otimes P_{i_2}H_{2,N}P_{i_1'}\otimes P_{i_2'}\Gamma_{\underline{J},\underline{J}'}^{(2)}\big),
	\end{equation}
	where a small amount of the kinetic energy has been sacrificed to bound the error terms. Fixing $i = i_{\textmd{max}}(\underline{J})$ and defining $\P_i = \sum_{k= 1}^iP_k$, we can rewrite the sum over $i_1,i_2,i_1',i_2'$ as
	\begin{equation*}
		\sum_{\substack{i_1,i_2 = 1\\ i_1',i_2' = 1}}^{i}\Tr\big(P_{i_1}\otimes P_{i_2}H_{2,N}P_{i_1'}\otimes P_{i_2'}\Gamma_{\underline{J},\underline{J}'}^{(2)}\big) = \Tr\big(\P_i^{\otimes 2}H_{2,N}\P_i^{\otimes 2}\Gamma_{\underline{J},\underline{J}'}^{(2)}\big),
	\end{equation*}
	which is almost of the right form to apply Theorem~\ref{th:de_finetti} on $\mfX = \P_i\mfH$. 
	
	The only remaining obstacle is that, instead of having $\Gamma_{\underline{J},\underline{J}'}$, which is in general not a even a state, we would like to have a state belonging to $\mathcal{S}(\mfX^K)$ for some $K$. In other words, we would like to know exactly how many particles have momenta in $\P_i$, and discard the information about the others. This can rigorously be done by not only fixing $i = i_{\textmd{max}}(\underline{J}),$ but also the number $K$ of particles that have momenta in $\P_i$, and using Lemma~\ref{lemma:state_one_projection_new_state} to construct a state $\gamma_{i,K} \in \mathcal{S}(\P_i^{\otimes K}\mfH^K)$. Having done so, we shall get
	\begin{equation}
		\label{eq:proof_strategy_lower_bound_2}
		\Tr\big(H_{2,N}\Gamma^{(2)}\big) \gtrsim \sum_{i = 1}^M\sum_K{N \choose 2}^{-1}{K \choose 2}\Vert\Psi_{i,K}\Vert^2\Tr\big(H_{2,N}\gamma_{i,K}^{(2)}\big),
	\end{equation}
	where $\Psi_{i,K}$ is given by
	\begin{equation*}
		\Psi_{i,K} = \sum_{\substack{\underline{J}\\ i_{\textmd{max}}(\underline{J}) = i\\ j_{\textmd{max}}(\underline{J}) = K}}\Psi_{\underline{J}} \quad \textmd{with} \quad j_{\textmd{max}}(\underline{J}) \coloneqq \sum_{k = 1}^{i_{\textmd{max}}(\underline{J})}j_k.
	\end{equation*}
	For each state $\gamma_{i,K}$, an application of Theorem~\ref{th:de_finetti} - which can be done since $\P_i\mfH$ has finite dimension controlled by the CLR type estimate \eqref{eq:clr_estimate} - shall yield
	\begin{equation}
		\label{eq:proof_strategy_de_finetti_lower_bound}
		\Tr\big(H_{2,N}\gamma_{i,K}^{(2)}\big) \gtrsim \int_{\mathcal{S}(\P_i\mfH)}\Tr\big(H_{2,N}\gamma^{\otimes 2}\big)\d{}\mu(\gamma) \gtrsim \tilde{e}_N^{\textmd{H}},
	\end{equation}
	where $\tilde{e}_N^{\textmd{H}}$ is the ground state energy of a slightly modified version of the Hartree functional \ref{eq:hartree_functional} that satisfies
	\begin{equation*}
		\tilde{e}_N^{\textmd{H}} \geq - C \quad \textmd{and} \quad \lim_{N\rightarrow\infty}\tilde{e}_N^{\textmd{H}} = e^\nls.
	\end{equation*}
	Finally, injecting \eqref{eq:proof_strategy_de_finetti_lower_bound} into \eqref{eq:proof_strategy_lower_bound_2} and controlling the error terms using Proposition~\ref{prop:plane_wave_estimate} shall yield the desired result.
	
	\begin{proposition}[Plane wave estimate]
		\label{prop:plane_wave_estimate}
		Let $\mathbf{e}_k$ be the multiplication operator on $\mfH$ by $\cos(k\cdot x)$ or $\sin(k\cdot x)$. Let $P_i$ denote the projection $\mathds{1}_{\{\sqrt{h} < N^{i\varepsilon}\}}$. Then, for all $p\in \mathbb{N}_0$ and $k\in\mathbb{R}^2\setminus\{0\}$,
		\begin{equation}
			\label{eq:plane_wave_estimate}
			\pm P_i\mathbf{e}_kP_i \leq C_p\dfrac{N^{pi\varepsilon}}{\vert k\vert^p}P_i,
		\end{equation}
		for $N$ large enough and for $C_p > 0$ depending only on $p$. Consequently, for all $i_1,i_2\in \{1,\dots,M\}$, we have
		\begin{equation}
			\label{eq:plane_wave_estimate_potential}
			P_{i_1}\otimes P_{i_2} \vert w_N(x-y)\vert P_{i_1}\otimes P_{i_2} \geq -C N^{2\min(i_1,i_2)\varepsilon} P_{i_1}\otimes P_{i_2},
		\end{equation}
		for $N$ large enough and for some constant $C>0$ (depending only on $\Vert w\Vert_{L^1}$).
	\end{proposition}
	
	\begin{lemma}
		\label{lemma:state_two_projections_full_trace}
		Let $P_1,P_2$ and $Q$ be orthogonal projections on $\mfH$. Given a state
		\begin{equation*}
			\Gamma\in \mathcal{S}\big(P_1^{\otimes j_1}\otimes_\textmd{s} P_2^{\otimes j_2}\otimes_\textmd{s} Q^{\otimes (N-j_1-j_2)}\mfH^N\big),
		\end{equation*}
		for some $0\leq j_1,j_2\leq N$, we have
		\begin{equation}
			\label{eq:state_one_projection_full_trace}
			\Tr\big(P_1\Gamma^{(1)}\big)= \dfrac{j_1}{N}, \quad \Tr\big(P_1^{\otimes 2}\Gamma^{(2)}\big) = \dfrac{j_1(j_1 - 1)}{N(N - 1)}
		\end{equation}
		and
		\begin{equation}
			\label{eq:state_two_projections_full_trace}
			\Tr\big(P_1\otimes P_2 \Gamma^{(2)}\big)= \dfrac{j_1j_2}{N(N-1)}.
		\end{equation}
	\end{lemma}
	
	\begin{lemma}
		\label{lemma:state_one_projection_new_state}
		Let $P$ and $Q$ be orthogonal projections on $\mfH$. Given a state
		\begin{equation*}
			\Gamma \in \mathcal{S}\big(P^{\otimes j}\otimes_{\textmd{s}}Q^{\otimes (N - j)}\mfH^N\big),
		\end{equation*}
		for some $j \geq 1$, there exists another state
		\begin{equation*}
			\Gamma_{j}\in\mathcal{S}\big(P^{\otimes j}\mfH^j\big)
		\end{equation*}
		such that
		\begin{equation}
			\label{eq:state_one_porjection_new_state}
			P\otimes P\Gamma^{(2)}P\otimes P = {N \choose 2}^{-1}{j \choose 2}\Gamma_{j}^{(2)}.
		\end{equation}
	\end{lemma}
	
	\noindent
	\textbf{Organisation of the paper.} We prove Theorem~\ref{th:convergence_nls_theory} in Section~\ref{section:proof_main_theorem}. Then, we prove Proposition~\ref{prop:plane_wave_estimate}, as well as Lemmas~\ref{lemma:state_two_projections_full_trace}~and~\ref{lemma:state_one_projection_new_state} in Section~\ref{section:main_estimate_other_results}. Lastly, in Section~\ref{section:proof_main_theorem_main_lemma} we prove an important technical lemma used in the proof of Theorem~\ref{th:convergence_nls_theory}.
	
	\section{Proof of Theorem~\ref{th:convergence_nls_theory}}
	
	\label{section:proof_main_theorem}
	
	We begin the proof by splitting $L^2(\mathbb{R}^2)$ into $M + 1$ annuli in momentum space according to the spectral projections
	\begin{equation}
		\label{eq:projections_momentum_space}
		P_1 = \mathds{1}_{\{h < N^{2\varepsilon}\}}, \quad P_i = \mathds{1}_{\{N^{2(i-1)\varepsilon} \leq h < N^{2i\varepsilon}\}}, \quad 2 \leq  i \leq M
	\end{equation}
	and
	\begin{equation}
		\label{eq:projections_momentum_space_largest}
		P_{M + 1} = \mathds{1}_{\{N^{2M\varepsilon}\leq h\}},
	\end{equation}
	for some $\varepsilon, M > 0$ that shall be chosen later. Moreover, we define $\P_j$ and $\Q_j$ by
	\begin{equation*}
		\P_j = \sum_{i = 1}^jP_i \quad \textmd{and} \quad \Q_j = \sum_{i = j + 1}^{M + 1}P_i.
	\end{equation*}
	In the polynomial case, that is for $w_N$ of the form \eqref{eq:interaction_potential_beta_scaling}, the parameter $\varepsilon$ shall be universal, and $M$ shall be a constant taken such that
	\begin{equation}
		\label{eq:relation_M_epsilon_poly}
		2\beta+\frac{1 + \eta}{\eta}\varepsilon < M\varepsilon < 4\beta + \frac{1 + \eta}{\eta}\varepsilon,
	\end{equation}
	where $\eta > 0$ is such that $w\in L^{1 + \eta}(\R^2)$. In the exponential case, that is $w_N$ of the form \eqref{eq:interaction_potential_kappa_scaling}, the parameter $\varepsilon$ shall depend only on $\kappa$. The parameter $M$ shall be taken such that
	\begin{equation}
		\label{eq:relation_M_epsilon_expo}
		\dfrac{2N^\kappa}{\log N} + \frac{1 + \eta}{\eta}\varepsilon < M\varepsilon < \dfrac{4N^\kappa}{\log N} + \frac{1 + \eta}{\eta}\varepsilon,
	\end{equation}			
	where $\eta > 0$ is again such that $w\in L^{1 + \eta}(\R^2)$. In both \eqref{eq:relation_M_epsilon_poly} and \eqref{eq:relation_M_epsilon_expo}, the lower bound arises when getting rid of the projection $P_{M+1}$, and the upper bound states that we do not want $M$ too large.
	
	Let $\Psi$ be the ground state of $H_N$, and $\Gamma$ be the projection $\vert\Psi\rangle\langle\Psi\vert$. Denote by $T$ the two-body operator $(h_1 + h_2)/2$. Define the two-body Hamiltonian
	\begin{equation*}
		H_{2,N} = T + w_N(x_1 - x_2)/2.
	\end{equation*}
	For readability's sake, we shall use the short-hand notation $w_N$ to denote the multiplication operator by $w_N(x_1 - x_2)$. By symmetry of $\Psi$ and $H_N$, we have
	\begin{equation*}
		e_N = \langle\Psi,H_N\Psi\rangle/N = \Tr\big(H_{2,N}\Gamma^{(2)}\big).
	\end{equation*}
	Since we need to extract a small amount of the kinetic energy to control some error terms, we define the modified Hamiltonian
	\begin{equation*}
		\tilde{H}_{2,N} = \left(1 - \xi_N\right)T + w_N/2,
	\end{equation*}
	with $\xi_N$ given by
	\begin{equation*}
		\label{eq:epsilon_kinetic_energy_condition}
		\xi_N = \xi + 1/\log N,
	\end{equation*}
	for some $\xi\in(0,1)$ to be determined. We claim that
	\begin{equation}
		\label{eq:hamiltonian_no_high_momenta_bound}
		\Tr\big(H_{2,N}\Gamma^{(2)}\big) \geq \Tr\big(\P_M^{\otimes 2}\tilde{H}_{2,N}\P_M^{\otimes 2}\Gamma^{(2)}\big) + \tilde{\xi}_N\Tr\big(h\Gamma^{(1)}\big) - CN^{-\varepsilon},
	\end{equation}
	with
	\begin{equation}
		\label{eq:epsilon_kinetic_energy_condition_new}
		\tilde{\xi}_N = \xi + C/\log N.
	\end{equation}
	In fact, using the identity $\mathds{1} = \P_M + P_{M + 1}$, we see that, to prove \eqref{eq:hamiltonian_no_high_momenta_bound}, we need to bound terms of the form
	\begin{equation*}
		\Tr\big(P_{M + 1}\otimes \P_Mw_N\P_M\otimes\P_M\Gamma^{(2)}\big),
	\end{equation*}
	and terms with more than one projection $P_{M + 1}$, which are bounded in the same way. Thanks to the Hölder inequality and the Sobolev inequality, we have
	\begin{multline*}
		\big\vert\Tr\big(P_{M + 1}\otimes \P_Mw_N\P_M\otimes\P_M\Gamma^{(2)}\big)\big\vert\\
		\begin{aligned}[t]
			&\leq CN^{-\varepsilon}\Tr\big(P_{M + 1}\otimes\P_M(N^{4\beta+2\varepsilon\frac{1 + \eta}{\eta}} + h_1)P_{M + 1}\otimes\P_M\Gamma^{(2)}\big)\\
			&\phantom{\leq} + CN^{-\varepsilon}\Tr\big(\P_M\otimes \P_M(1 + h_1)\P_M\otimes \P_M\Gamma^{(2)}\big)\\
			&\leq CN^{-\varepsilon}\Tr\big(P_{M + 1}\otimes \P_M(1 + h_1)P_{M + 1}\otimes \P_M\Gamma^{(2)}\big)\\
			&\phantom{\leq} + CN^{-\varepsilon}\Tr\big(\P_M\otimes \P_M(1 + h_1)\P_M\otimes\P_M\Gamma^{(2)}\big).
		\end{aligned}
	\end{multline*}
	In the first inequality we used the Sobolev inequality with weights
	\begin{equation*}
		\Vert f\Vert_{L^p}^2 \leq C_p\tau^{-(p-2)/p}\left(\Vert f\Vert^2 + \tau\Vert\nabla f\Vert^2\right), \quad \forall f\in H^1(\R^2),
	\end{equation*}
	which holds for any $p \geq 2$ and for all $\tau > 0$, and that can be proven following the same proof as \cite[Theorem 8.5]{Lieb2001Analysis}. In the second inequality we used \eqref{eq:relation_M_epsilon_poly} (or \eqref{eq:relation_M_epsilon_expo}) and the estimate $N^{2M\varepsilon}P_{M + 1} \leq P_{M + 1}hP_{M + 1}$. This proves \eqref{eq:hamiltonian_no_high_momenta_bound} for $N$ large enough.
	
	After that, we decompose $\Psi$ according to the occupancy of the energy levels defined by the projections \eqref{eq:projections_momentum_space} and \eqref{eq:projections_momentum_space_largest}. For this purpose, we define the function $\Psi_{\underline{J}}$ by
	\begin{equation*}
		\Psi_{\underline{J}} = P_1^{\otimes j_1}\otimes_{\s} \cdots\otimes_{\s} P_{M+1}^{\otimes j_{M+1}}\Psi,
	\end{equation*}
	for any multi-index $\underline{J} = (j_1,\dots,j_{M + 1})$ satisfying $\vert\underline{J}\vert = N$. Define also $\Gamma_{\underline{J}}$ and $\Gamma_{\underline{J},\underline{J}'}$ by
	\begin{equation*}
		\Gamma_{\underline{J}} = \left\vert\Psi_{\underline{J}}\right\rangle\left\langle\Psi_{\underline{J}}\right\vert \quad \textmd{and} \quad \Gamma_{\underline{J},\underline{J}'} = \left\vert\Psi_{\underline{J}'}\right\rangle\left\langle\Psi_{\underline{J}}\right\vert.
	\end{equation*}
	Moreover, for any multi-index $\underline{J} = (j_1,\dots,j_{M + 1})$, we define
	\begin{equation*}
		\underline{J}^{(i_1,i_2)} \coloneqq
		\left\{
		\begin{array}{ll}
			(j_1,\dots,j_{i_1} + 1,\dots,j_{i_2} + 1, \dots, j_{M + 1}) & \textmd{if $i_1 \neq i_2$,}\\
			(j_1,\dots,j_{i_1} + 2, \dots, j_{M + 1}) & \textmd{otherwise}
		\end{array}
		\right.
	\end{equation*}
	as well as
	\begin{equation*}
		\Gamma_{\underline{J}^{(i_1,i_2;i_1',i_2')}} \coloneqq \Gamma_{\underline{J}^{(i_1,i_2)},\underline{J}^{(i_1',i_2')}} = \vert\Psi_{\underline{J}^{(i_1',i_2')}}\rangle\langle\Psi_{\underline{J}^{(i_1,i_2)}}\vert.
	\end{equation*}
	Then, using that
	\begin{equation*}
		P_{i_1'}\otimes P_{i_2'}\Gamma^{(2)}P_{i_1}\otimes P_{i_2} = \sum_{\underline{J}}P_{i_1'}\otimes P_{i_2'}\Gamma_{\underline{J}^{(i_1,i_2;i_1',i_2')}}^{(2)}P_{i_1}\otimes P_{i_2},
	\end{equation*}
	we can write
	\begin{align*}
		\Tr\big(H_{2,N}\Gamma^{(2)}\big) &\geq \sum_{\underline{J}}\sum_{\substack{i_1, i_2 = 1\\ i_1', i_2' = 1}}^M\Tr\big(P_{i_1}\otimes P_{i_2}\tilde{H}_{2,N}P_{i_1'}\otimes P_{i_2'}\Gamma_{\underline{J}^{(i_1,i_2;i_1',i_2')}}^{(2)}\big) + \tilde{\xi}_N\Tr\big(h\Gamma^{(1)}\big)\\
		&\phantom{\geq} - CN^{-\varepsilon}, \numberthis \label{eq:hamiltonian_two_particle_split}
	\end{align*}
	where we are summing over multi-indices $\underline{J}$ satisfying $\vert\underline{J}\vert = N - 2$.
	
	Let us bound the first term in the right-hand side of \eqref{eq:hamiltonian_two_particle_split}. To do so, we define
	\begin{equation}
		\label{eq:index_max_def}
		i_{\textmd{max}}(\underline{J}) = \max\left\{i\in\{1,\dots,M\}:j_i \geq N^{1- \delta\varepsilon}\right\},
	\end{equation}
	for some $\delta > 0$ that shall be fixed later.  We shall sometimes omit $\underline{J}$ and simply write $i_{\textmd{max}}$ for readability's sake. We take the convention that $i_{\textmd{max}}(\underline{J}) = 0$ if the set on the right-hand side is empty. Next, we make the following important distinction regarding the sum over $i_1,i_2,i_1',i_2'$: either $i_{\textmd{max}}(\underline{J}) = i_{\textmd{max}}(\underline{J}')$ and $i_1,i_2,i_1',i_2'\leq i_{\textmd{max}}(\underline{J})$ - in which case we shall use Theorem~\ref{th:de_finetti} - or not. In the latter case, we need to bound terms of the form
	\begin{equation*}
		\Tr\big(P_{i_1}\otimes P_{i_2}w_NP_{i_1'}\otimes P_{i_2'}\Gamma_{\underline{J}^{(i_1,i_2;i_1',i_2')}}^{(2)}\big),
	\end{equation*}
	where one of the four indices, say $i_1$, is such that $j_{i_1} < N^{1 - \delta\varepsilon}$. This is done in the following lemma, whose proof is given in Section~\ref{section:proof_main_theorem_main_lemma}.
	\begin{lemma}
		\label{lemma:bound_interaction_low_occupancy}
		Take $w_N$ as in Assumption~\ref{ass:potentials}, and $M$ and $\varepsilon$ satisfying either \eqref{eq:relation_M_epsilon_poly} or \eqref{eq:relation_M_epsilon_expo}. Let $\delta \geq 6$ and define $\imax$ as in \eqref{eq:index_max_def} for any given multi-index $\underline{J}$. Then,
		\begin{multline}
			\label{eq:bound_interaction_low_occupancy}
			\sum_{J}\sum_{i_1=i_{\max}+1}^M\sum_{i_2,i_3,i_4=1}^M\big\vert\Tr\big(P_{i_1}\otimes P_{i_2}w_NP_{i_1'}\otimes P_{i_2'}\Gamma_{\underline{J}^{(i_1,i_2;i_1',i_2')}}^{(2)}\big)\big\vert\\
			\leq C(N^{-\varepsilon} + N^{2\varepsilon - (1 - \kappa)})\Tr\big((1 + h)\Gamma^{(1)}\big)
		\end{multline}
		for some constant $C$ that depends only on $\left\Vert w\right\Vert_1$, and with $\kappa = 0$ in the polynomial case \eqref{eq:relation_M_epsilon_poly}. Here, we are summing over all multi-indices $\underline{J}$ satisfying $\vert\underline{J}\vert = N - 2$.
	\end{lemma}
	Injecting \eqref{eq:bound_interaction_low_occupancy} into \eqref{eq:hamiltonian_two_particle_split}, we obtain
	\begin{align*}
		\Tr\big(H_{2,N}\Gamma^{(2)}\big) &\geq \sum_{\underline{J}}\sum_{\substack{i_1, i_2 = 1\\ i_1', i_2' = 1}}^{i_{\textmd{max}}}\Tr\big(P_{i_1}\otimes P_{i_2}\tilde{H}_{2,N}P_{i_1'}\otimes P_{i_2'}\Gamma_{\underline{J}^{(i_1,i_2;i_1',i_2')}}^{(2)}\big) + \tilde{\xi}_N\Tr\big(h\Gamma^{(1)}\big)\\
		&\phantom{\geq} - CN^{-\varepsilon} - CN^{2\varepsilon -(1-\kappa)}, \numberthis
	\end{align*}
	for some new $\tilde{\xi}_N$ that is still of the form \eqref{eq:epsilon_kinetic_energy_condition_new}, and under the condition
	\begin{equation}
		\label{eq:condition_epsilon_kappa}
		\varepsilon < \dfrac{1 - \kappa}{2},
	\end{equation}
	with $\kappa = 0$ in the polynomial case \eqref{eq:interaction_potential_beta_scaling}. We do not mention the condition \eqref{eq:condition_epsilon_kappa} further since it will be fulfilled by our choice of $\varepsilon$.
	
	Before we can apply Theorem~\ref{th:de_finetti}, we need to do a bit of rewriting. We start by decomposing over the different values that $\imax$ can take:
	\begin{equation*}
		\sum_{\underline{J}}\sum_{\substack{i_1, i_2 = 1\\ i_1', i_2' = 1}}^{i_{\textmd{max}}}P_{i_1'}\otimes P_{i_2'}\Gamma_{\underline{J}^{(i_1,i_2;i_1',i_2')}}^{(2)}P_{i_1}\otimes P_{i_2} = \sum_{i=1}^M\sum_{\substack{i_1, i_2 = 1\\ i_1', i_2' = 1}}^i\sum_{\substack{\underline{J}\\ \imax(\underline{J}) = i}}P_{i_1'}\otimes P_{i_2'}\Gamma_{\underline{J}^{(i_1,i_2;i_1',i_2')}}^{(2)}P_{i_1}\otimes P_{i_2}.
	\end{equation*}
	Then, we use that, at fixed $i_1,i_2,i_1',i_2'$, the sum over $\underline{J}$ can be rewritten as
	\begin{equation}
		\label{eq:two_prdm_mutli_indices_rewritten}
		\sum_{\substack{\underline{J}\\ \imax(\underline{J}) = i}}P_{i_1'}\otimes P_{i_2'}\Gamma_{\underline{J}^{(i_1,i_2;i_1',i_2')}}^{(2)}P_{i_1}\otimes P_{i_2} = \sum_{\substack{\underline{J}',\underline{J}''\\ \imax(\underline{J}') = i\\ \imax(\underline{J}'') = i}}P_{i_1'}\otimes P_{i_2'}\Gamma_{\underline{J}',\underline{J}''}^{(2)}P_{i_1}\otimes P_{i_2},
	\end{equation}
	where we are summing over $\underline{J}$'s satisfying $\vert\underline{J}\vert = N - 2$ on the left-hand side, and over $\underline{J}',\underline{J}''$ satisfying $\vert\underline{J}'\vert = \vert\underline{J}''\vert = N$ on the right-hand side. Notice first that, due to the presence of the projections $P_{i_1}\otimes P_{i_2}$ and $P_{i_1'}\otimes P_{i_2'}$, the multi-index $\underline{J}''$ is entirely determined by the choice of $\underline{J}'$. Said differently, the double sum on the right-hand side is in essence only a single sum. Moreover, for every $\underline{J}'$, there is exactly one $\underline{J}$ such that $\underline{J}^{(i_1,i_2)} = \underline{J}'$. Combining the previous two arguments is enough to prove \eqref{eq:two_prdm_mutli_indices_rewritten}. We now define $j_{\textmd{max}}(\underline{J})$ by
	\begin{equation*}
		j_{\textmd{max}}(\underline{J}) = \sum_{i = 1}^{i_{\textmd{max}}(\underline{J})}j_i,
	\end{equation*}
	as well as $\Psi_{i,K}$ and $\Gamma_{i,K}$ by
	\begin{equation}
		\label{eq:wavefunction_decomposed_i_K}
		\Psi_{i,K} =  \sum_{\substack{\underline{J}\\ i_{\textmd{max}}(\underline{J}) = i\\ j_{\textmd{max}}(\underline{J}) = K}}\Psi_{\underline{J}} \quad \textmd{and} \quad \Gamma_{i,K} = \left\vert\Psi_{i,K}\right\rangle\left\langle\Psi_{i,K}\right\vert.
	\end{equation}
	As it shall be useful later, we use the convention that, when $i = 0$, the condition $j_{\textmd{max}} = K$ in the definition of $\Psi_{0,K}$ \eqref{eq:wavefunction_decomposed_i_K} should be replaced by $j_{M + 1} = K$. With these new definitions, we may rewrite \eqref{eq:two_prdm_mutli_indices_rewritten} as
	\begin{equation*}
		\sum_{\substack{\underline{J}\\ \imax(\underline{J}) = i}}P_{i_1}\otimes P_{i_2}\Gamma_{\underline{J}^{(i_1,i_2;i_1',i_2')}}^{(2)}P_{i_1'}\otimes P_{i_2'} = \sum_{K}P_{i_1}\otimes P_{i_2}\Gamma_{i,K}^{(2)}P_{i_1'}\otimes P_{i_2'}.
	\end{equation*}
	Summing up, we have
	\begin{equation}
		\Tr\big(H_{2,N}\Gamma^{(2)}\big) \geq \sum_{i =1}^M\sum_{K}\Tr\big(\P_i^{\otimes 2}\tilde{H}_{2,N}\P_i^{\otimes 2}\Gamma_{i,K}^{(2)}\big) + \tilde{\xi}_N\Tr\big(h\Gamma^{(1)}\big) - CN^{-\varepsilon} - CN^{2\varepsilon-(1-\kappa)}, \label{eq:hamiltonian_lower_bound_before_de_finetti}
	\end{equation}
	where we are summing over $N - MN^{1 - \delta\varepsilon}\leq K\leq N$.
	
	Given that $\Psi_{i,K}$ is in $\P_i^{\otimes K}\otimes_{\textmd{s}}\Q_{i}^{\otimes (N - K)}\mfH^N$ and is symmetric, Lemma~\ref{lemma:state_one_projection_new_state} allows us to find a symmetric state $\gamma_{i,K}\in\mathcal{S}(\P_i^{\otimes K}\mfH^{K})$ such that
	\begin{equation*}
		\P_i^{\otimes 2}\Gamma_{i,K}^{(2)}\P_i^{\otimes 2} = {N \choose 2}^{-1}{K \choose 2}\left\Vert\Psi_{i,K}\right\Vert^2\gamma_{i,K}^{(2)}.
	\end{equation*}
	Then, writing $w_N(x-y)$ in Fourier, applying Theorem~\ref{th:de_finetti} and using Proposition~\ref{prop:plane_wave_estimate} and the CLR type estimate \eqref{eq:clr_estimate} to bound the error terms, we obtain
	\begin{align*}
		\Tr\big(\tilde{H}_{2,N} \gamma_{i,K}^{(2)}\big)
		&= \int_{\mathcal{S}(\P_i\mfH)}\Tr\big(\tilde{H}_{2,N}\gamma^{\otimes 2}\big)\d{}\mu_{i,K}(\gamma)\\
		&\phantom{=} + \left(1 - \xi_N\right)\Tr\left(T\left[\gamma_{i,K}^{(2)} - \int_{\mathcal{S}(\P_i\mfH)}\gamma^{\otimes 2}\d{}\mu_{i,K}(\gamma)\right]\right)\\
		&\phantom{=} + \dfrac{1}{2}\int_{\mathbb{R}^2}\d{}k\hat{w}(N^{-\beta}k)\Tr\left(e^{\mathbf{i}k\cdot x}e^{-\mathbf{i}k\cdot y}\left[\gamma_{i,K}^{(2)} - \int_{\mathcal{S}(\P_i\mfH)}\gamma^{\otimes 2}\d{}\mu_{i,K}(\gamma)\right]\right)\\
		&\geq \int_{\mathcal{S}(\P_i\mfH)}\Tr\big(\tilde{H}_{2,N}\gamma^{\otimes 2}\big)\d{}\mu_{i,K}(\gamma)\\
		&\phantom{\geq} - C\sqrt{\dfrac{M\varepsilon\log N}{K}}\left(N^{2i\varepsilon} + \int_{\mathbb{R}^2}\d{}k\Vert \P_ie^{\mathbf{i}k\cdot x}\P_i\Vert_{\textmd{op}}^2\vert \hat{w}(N^{-\beta}k)\vert\right)\\
		&\geq \int_{\mathcal{S}(\P_i\mfH)}\Tr\big(\tilde{H}_{2,N}\gamma^{\otimes 2}\big)\d{}\mu_{i,K}(\gamma) - C\sqrt{\dfrac{M\varepsilon\log N}{K}}N^{2i\varepsilon}, \numberthis \label{eq:hamiltonian_estimate_de_finetti}
	\end{align*}
	where $\mu_{i,K}$ is the probability measure from Theorem~\ref{th:de_finetti} applied to the state $\gamma_{i,K}$. To prove the last inequality, we decomposed the integral over $k$ just as we did in the proof of Proposition~\ref{prop:plane_wave_estimate}. As a result of \eqref{eq:hamiltonian_lower_bound_before_de_finetti} and \eqref{eq:hamiltonian_estimate_de_finetti}, we have
	\begin{align*}
		\Tr\big(H_{2,N}\Gamma^{(2)}\big) &\geq \sum_{i = 1}^M\sum_K{N \choose 2}^{-1}{K \choose 2}\left\Vert\Psi_{i,K}\right\Vert^2\int_{\mathcal{S}(\P_i\mfH)}\Tr\big(\tilde{H}_{2,N}\gamma^{\otimes 2}\big)\d{}\mu_{i,K}(\gamma)\\
		&\phantom{\geq} - C\sum_{i=1}^M\sum_K{N \choose 2}^{-1}{K \choose 2}\left\Vert\Psi_{i,K}\right\Vert^2\sqrt{\dfrac{M\varepsilon\log N}{K}}N^{2i\varepsilon}\\
		&\phantom{\geq} + \tilde{\xi}_N\Tr\big(h\Gamma^{(1)}\big) - CN^{-\varepsilon} - CN^{2\varepsilon - (1 - \kappa)}. \numberthis
		\label{eq:hamiltonian_estimate_de_finetti_after}
	\end{align*}
	On the one hand, using ${N \choose 2}^{-1}{K \choose 2} \leq 1$, we can bound the second term in \eqref{eq:hamiltonian_estimate_de_finetti_after} by
	\begin{multline*}
		N^{\delta\varepsilon/2 - 1/2}\sqrt{M\varepsilon\log N}\sum_{i = 1}^M\sum_K{N \choose 2}^{-1}{K \choose 2}\left\Vert \Psi_{i,K}\right\Vert^2N^{2i\varepsilon}\\
		\begin{aligned}[t]
			&\leq N^{\delta\varepsilon/2 - 1/2}\sqrt{M\varepsilon\log N}\sum_{\underline{J}}\left\Vert \Psi_{\underline{J}}\right\Vert^2N^{2i_{\textmd{max}}(\underline{J})\varepsilon}\\
			&\leq CN^{(3\delta/2 + 2)\varepsilon - 1/2}\sqrt{M\varepsilon\log N}\sum_{\underline{J}}\left\Vert \Psi_{\underline{J}}\right\Vert^2N^{2(i_{\textmd{max}}(\underline{J}) - 1)\varepsilon}\dfrac{j_{i_{\textmd{max}}(\underline{J})}}{N}\\
			&\leq CN^{(3\delta/2 + 2)\varepsilon - 1/2}\sqrt{M\varepsilon\log N}\Tr\big(h\Gamma^{(1)}\big).
		\end{aligned}
	\end{multline*}
	Choosing $\delta = 6$ and $\varepsilon < 1/22$ in the polynomial case \eqref{eq:interaction_potential_beta_scaling}, and $\delta = 6$ and $\varepsilon < (1 - \kappa)/22$ in the exponential case \eqref{eq:interaction_potential_kappa_scaling}, we see that the previous term can be absorbed by the third term in the the right-hand side of \eqref{eq:hamiltonian_estimate_de_finetti_after}. On the other hand, to simplify the first term in \eqref{eq:hamiltonian_estimate_de_finetti_after} we define the Borel measure
	\begin{equation}
		\label{eq:de_finetti_borel_measure}
		\mu_N = \sum_{i = 1}^M\sum_K{N \choose 2}^{-1}{K \choose 2}\Vert\Psi_{i,K}\Vert^2\mu_{i,K},
	\end{equation}
	and the modified Hartree energy functional
	\begin{equation}
		\label{eq:hartree_functional_modified}
		\tilde{\cE}_N^{\textmd{H}}[\gamma] = (1 - \xi_N)\Tr(h\gamma) + \dfrac{1}{2}\Tr(w_N\gamma^{\otimes 2}).
	\end{equation}
	
	Summing up, the inequality \eqref{eq:hamiltonian_estimate_de_finetti_after} simplifies to
	\begin{equation}
		\label{eq:hamiltonian_two_particle_split_last}
		\Tr\big(H_{2,N}\Gamma^{(2)}\big) \geq \int_{\mathcal{S}(\P_M\mfH)}\tilde{\cE}_N^{\textmd{H}}[\gamma]\d{}\mu_N(\gamma) + \tilde{\xi}_N\Tr\big(h\Gamma^{(1)}\big) - CN^{-\varepsilon} - CN^{2\varepsilon - (1 - \kappa)},
	\end{equation}
	with $\mu_N$ given by \eqref{eq:de_finetti_borel_measure}, $\tilde{\cE}_N^{\textmd{H}}$ by \eqref{eq:hartree_functional_modified}, and $\tilde{\xi}_N$ of the form
	\begin{equation*}
		\tilde{\xi}_N = \xi + C/\log N,
	\end{equation*}
	for some $\xi\in(0,1)$ to be determined. To conclude the proof of Theorem~\ref{th:convergence_nls_theory}, we wish to take $\xi = 0$ and use that
	\begin{equation}
		\label{eq:hartree_functional_modified_liminf}
		\liminf_{N \rightarrow \infty}\int_{\mathcal{S}(\P_M\mfH)}\tilde{\cE}_N^{\textmd{H}}[\gamma]\d{}\mu_N(\gamma) \geq e^{\textmd{nls}},
	\end{equation}
	which we now show.
	
	Before proving \eqref{eq:hartree_functional_modified_liminf}, we show that the estimate \eqref{eq:hamiltonian_two_particle_split_last} implies the boundedness of the kinetic energy:
	\begin{equation}
		\label{eq:kinetic_energy_bounded}
		\Tr\big(h\Gamma^{(1)}\big) \leq C.
	\end{equation}
	Thanks to the convexity of the kinetic energy \cite[Theorem 7.8]{Lieb2001Analysis} and the Cauchy--Schwarz inequality, we have
	\begin{equation*}
		\tilde{\cE}_N^{\mathrm{H}}[\gamma] \geq (1 - \xi_N)\int_{\R^2}\vert\nabla\sqrt{\rho_\gamma}\vert^2 - \dfrac{1}{2}\int w_-\int_{\R^2}\rho_\gamma^2,
	\end{equation*}
	where $\rho_\gamma$ denotes the density $\gamma(x;x)$ (defined properly by the spectral decomposition of $\gamma$).	Using the Gagliardo--Nirenberg inequality \eqref{eq:gagliardo_nirenberg_inequality} and the assumption \eqref{eq:interaction_potential_negative_assumption}, we deduce that $\tilde{\cE}_N^{\mathrm{H}}[\gamma]$ is nonnegative for $\xi$ small enough and $N$ large enough. Combining this with the boundedness from above of $e_N$, which follows from a simple trial argument, we obtain \eqref{eq:kinetic_energy_bounded}.
	
	As a consequence of \eqref{eq:kinetic_energy_bounded}, the measure $\mu_N$ satisfies
	\begin{equation}
		\label{eq:de_finetti_borel_measure_almost_proba}
		1 - CN^{-\delta\varepsilon} - CN^{-2M\varepsilon} \leq \int_{\mathcal{S}(\P_M\mfH)}\d{}\mu_N(\gamma) \leq 1.
	\end{equation}
	The upper bound follows from the estimate ${N \choose 2}^{-1}{K \choose 2} \leq 1$ and
	\begin{equation*}
		\sum_{i = 0}^{M+1}\sum_K\Vert\Psi_{i,K}\Vert^2 = 1.
	\end{equation*}
	For the lower bound, we use
	\begin{align*}
		\int_{\mathcal{S}(\P_M\mfH)}\d{}\mu_N(\gamma) = \sum_{i = 1}^M\sum_K\Tr\big(\P_i^{\otimes 2}\Gamma_{i,K}^{(2)}\big) &= \sum_{i = 0}^{M + 1}\sum_K\left\Vert\Psi_{i,K}\right\Vert^2 - \sum_{K}\left\Vert\Psi_{0,K}\right\Vert^2\\
		&\phantom{=} - \sum_{i = 1}^M\sum_K\big(\Tr\big(\Q_{i}\Gamma_{i,K}^{(1)}\big) + \Tr\big(\Q_i\otimes\P_i\Gamma_{i,K}^{(2)}\big)\big).
	\end{align*}
	The first term in the right is equal to $1$. The last two can be bounded using Lemma~\ref{lemma:state_two_projections_full_trace}, and the remaining one satisfies
	\begin{equation*}
		\sum_K\Vert\Psi_{0,K}\Vert^2 \leq CN^{-2M\varepsilon}\Tr\big(h\Gamma^{(1)}\big) \leq CN^{-2M\varepsilon}.
	\end{equation*}
	This proves \eqref{eq:de_finetti_borel_measure_almost_proba}.
	
	For purely attractive potentials, i.e. $w \leq 0$, we can conclude the proof of Theorem~\ref{th:convergence_nls_theory} rather easily. Take $\xi = 0$ and define the modified NLS energy functional $\tilde{\cE}_N^\nls$ by
	\begin{equation*}
		\tilde{\cE}_N^\nls[u] = (1 - \xi_N)\langle u,hu\rangle - \dfrac{1}{2}\int_{\R^2}w_-\int_{\R^2}\vert u\vert^4,
	\end{equation*}
	and let $\tilde{e}_N^\nls$ denote its ground state energy. Then, we write the mixed state $\gamma$ in \eqref{eq:hamiltonian_two_particle_split_last} as
	\begin{equation*}
		\gamma = \sum_{j}\lambda_j\vert u_j\rangle\langle u_j\vert,
	\end{equation*}
	for some orthonormal family $(u_j)_j$, and use the Cauchy--Schwarz inequality to deduce
	\begin{equation}
		\label{eq:energy_functional_nls_mixed_state_attractive_case}
		\tilde{\cE}_N^\nls \geq \sum_j\lambda_j\tilde{\cE}_N^\nls[u_j] \geq \tilde{e}_N^\nls,
	\end{equation}
	The convergence of $\tilde{e}_N^\nls$ to $e^\nls$ (see \cite[Section 4]{Lewin2014TheMA} and \cite[Appendix A]{Lewin2017Note2D}) and the estimate \eqref{eq:de_finetti_borel_measure_almost_proba} yields \eqref{eq:hartree_functional_modified_liminf}.  Combining this with \eqref{eq:hamiltonian_two_particle_split_last} yields the convergence of the energies in Theorem~\ref{th:convergence_nls_theory}. Stability of second kind follows from the boundedness from below of $\tilde{e}_N^\nls$.

	When $w$ is not purely attractive, we cannot use the simple argument \eqref{eq:energy_functional_nls_mixed_state_attractive_case}, and we need to prove that the sequence of measures $(\mu_N)_N$ converges to a probability measure supported on pure states $\vert u\rangle\langle u\vert$. More or less the same proof as in \cite[Section 3.4.]{Rougerie2020NLSagain} applies in our case (relying on \eqref{eq:kinetic_energy_bounded} and \eqref{eq:de_finetti_borel_measure_almost_proba}); we do not repeat the arguments here to avoid cluttering the proof. Once this is shown, the conclusion of Theorem~\ref{th:convergence_nls_theory} is essentially the same as in the purely attractive case.
	
	\section{Plain wave estimate and other results}
	
	\label{section:main_estimate_other_results}
	
	\begin{proof}[Proof of Proposition~\ref{prop:plane_wave_estimate}]
		Take $k\in\mathbb{R}^2\setminus\{0\}$ and let us prove \eqref{eq:plane_wave_estimate} for $\mathbf{e}_k = \cos(k\cdot x)$ ($\mathbf{e}_k = \sin(k\cdot x)$ being similar). The bound for $p = 0$ is just $\vert\mathbf{e}_k\vert \leq 1$. For $p = 1$, we take some smooth compactly supported function $f$ and write
		\begin{align*}
			\langle P_if, \cos(k\cdot x)P_if\rangle &= \int_{\mathbb{R}^2}\vert P_if\vert^2\cos(k\cdot x)\d{}x = -\int\nabla \vert P_if\vert^2\frac{k \sin(k\cdot x)}{\vert k\vert^2}\d{}x\\
			&\leq \frac{2}{\vert k\vert}\int_{\mathbb{R}^2}\big\vert \nabla \vert P_if\vert \big\vert \vert P_if\vert \leq \frac{2}{\vert k\vert}\Vert\nabla P_if\Vert\Vert P_if\Vert \leq C\frac{N^{i\varepsilon}}{\vert k\vert}\Vert P_if\Vert^2,
		\end{align*}
		for $N$ large enough. Here, we wrote $\cos(k\cdot x) = k\cdot \nabla\sin(k\cdot x)/\vert k\vert^2$ and used an integration by parts. Similarly, to prove \eqref{eq:plane_wave_estimate} for $p \geq 2$ we do $p$ integration by parts instead of a single one.
		
		To prove \eqref{eq:plane_wave_estimate_potential}, we first decompose $w_N$ in Fourier space
		\begin{align*}
			w_N(x-y) &= \int_{\mathbb{R}^2}\d{}k \hat{w}(N^{-\beta}k)e^{\mathbf{i}k\cdot(x-y)}\\
			&= \int_{\mathbb{R}^2}\d{}k \hat{w}(N^{-\beta}k)\left(\cos(k\cdot x)\cos(k\cdot y)+\sin(k\cdot x)\sin(k\cdot y)\right). \numberthis \label{eq:plane_wave_decomposition_potential}
		\end{align*}
		Conjugating the previous expression by $P_{i_1}\otimes P_{i_2}$, we see that we must bound terms of the form $P_{i}\mathbf{e}_kP_{i}$. Decomposing the integral \eqref{eq:plane_wave_decomposition_potential} into integrals over $\{|k| \leq N^{i_1\varepsilon}\}$, $\{N^{i_1\varepsilon} < |k|\}$ (taking $i_1 \leq i_2$), and using the estimate \eqref{eq:plane_wave_estimate} with $p = 0$ for the first integral and $p = 3$ for the second one, we obtain
		\begin{align*}
			P_{i_1}\otimes P_{i_2} w_N(x-y)P_{i_1}\otimes P_{i_2} &\geq -\int_{\vert k\vert \leq N^{i_1\varepsilon}}\d{}k\vert \hat{w}(N^{-\beta}k)\vert P_{i_1}\otimes P_{i_2}\\
			&\phantom{\geq} -\int_{N^{i_1\varepsilon} <\vert k \vert}\d{}k\vert \hat{w}(N^{-\beta}k)\vert\dfrac{N^{3i_1\varepsilon}}{|k|^3} P_{i_1}\otimes P_{i_2}.
		\end{align*}
		Bounding $\vert \hat{w}\vert$ by $\|\hat{w}\|_{L^\infty} \leq \Vert w\Vert_{L^1}$ and integrating over $k$ directly implies \eqref{eq:plane_wave_estimate_potential}, thereby finishing the proof of Proposition~\ref{prop:plane_wave_estimate}.
	\end{proof}
	
	We now prove Lemmas~\ref{lemma:state_two_projections_full_trace}~and~\ref{lemma:state_one_projection_new_state}. Define the permutation operator $U_\sigma: \mfH^N \rightarrow \mfH^N$, for some $\sigma\in S_N$, acting as
	\begin{equation*}
		U_\sigma\Psi(x_1,\dots,x_N) = \Psi(x_{\sigma(1)},\dots,_{\sigma(N)}),
	\end{equation*}
	for all $\Psi\in\mfH^N$. Then, that a state $\Gamma\in\mathcal{S}(\mfH^N)$ is symmetric can be restated as
	\begin{equation}
		\label{eq:state_symmetric_restated}
		U_\sigma\Gamma U_\sigma^* = \Gamma,.
	\end{equation}
	for all permutations $\sigma\in S_N$. Moreover, for any family $\{P_1,\dots,P_k\}$ of orthogonal projections on $\mfH$, we can write
	\begin{equation}
		\label{eq:symmetric_tensor_product_orthogonal_projections_rewritten}
		P_1^{\otimes n_1}\otimes_\textmd{s}\dots\otimes_\textmd{s}P_k^{\otimes n_k} = \dfrac{1}{n_1!\cdots n_k!}\sum_{\sigma\in S_N}U_\sigma P_1^{\otimes n_1}\otimes\dots\otimes P_k^{\otimes n_k}U_\sigma^*,
	\end{equation}
	for any nonnegative integers $n_1,\dots,n_k$ such that $\sum_{i = 1}^kn_i = N$.
	
	\begin{proof}[Proof of Lemma~\ref{lemma:state_two_projections_full_trace}]
		We only prove \eqref{eq:state_two_projections_full_trace} since \eqref{eq:state_one_projection_full_trace} follows similarly. We begin by writing
		\begin{equation*}
			\Tr\big(P_1\otimes P_2\Gamma^{(2)}\big) = \Tr\big(P_1\otimes P_2\otimes \mathds{1}^{\otimes(N - 2)}\Gamma\big) = \Tr\big(P_1\otimes P_2\otimes(P_1 + P_2 + Q)^{\otimes(N - 2)}\Gamma\big).
		\end{equation*}
		Developing $(P_1 + P_2 + Q)^{\otimes(N - 2)}$ and using that, since $\Gamma$ acts on $P_1^{\otimes j_1}\otimes_\textmd{s}P_2^{\otimes j_2}\otimes_\textmd{s}Q^{\otimes (N - j_1 - j_2)}\mfH^N$, the only terms that do not vanish are the ones containing $j_1 - 1$ times $P_1$, $j_2 - 1$ times $P_2$ and $N - j_1 - j_2$ times $Q$, we obtain
		\begin{equation*}
			\Tr\big(P_1\otimes P_2\Gamma^{(2)}\big) = \Tr\big(P_1\otimes P_2\otimes \big(P_1^{\otimes(j_1 - 1)}\otimes_\textmd{s}P_2^{\otimes(j_2 - 1)}\otimes_\textmd{s}Q^{\otimes(N - j_1 - j_2)}\big)\Gamma\big).
		\end{equation*}
		Thanks to \eqref{eq:state_symmetric_restated}, this becomes
		\begin{equation*}
			\Tr\big(P_1\otimes P_2\Gamma^{(2)}\big) = \dfrac{1}{N!}\sum_{\sigma\in S_N}\Tr\big(U_\sigma^*P_1\otimes P_2\otimes\big(P_1^{\otimes(j_1 - 1)}\otimes_\textmd{s}P_2^{\otimes (j_2 - 1)}\otimes_\textmd{s}Q^{\otimes(N - j_1 - j_2)}\big)U_\sigma\Gamma\big).
		\end{equation*}
		Then, using the cyclicity of the trace as well as \eqref{eq:symmetric_tensor_product_orthogonal_projections_rewritten}, we find
		\begin{align*}
			\Tr\big(P_1\otimes P_2\Gamma^{(2)}\big)
			&= 
			\begin{multlined}[t]
				\dfrac{1}{N!(j_1 - 1)!(j_2 - 1)!(N - j_1 - j_2)!}\sum_{\sigma\in S_N}\sum_{\pi\in S_{N - 2}}\\
				\Tr\big(U_\sigma^*P_1\otimes P_2\otimes U_\pi^*P_1^{\otimes(j_1 - 1)}\otimes P_2^{\otimes(j_2 - 1)}\otimes Q^{\otimes(N - j_1 - j_2)}U_\pi U_\sigma\Gamma\big)
			\end{multlined}\\
			&= 
			\begin{multlined}[t]
				\dfrac{(N - 2)!}{N!(j_1 - 1)!(j_2 - 1)!(N - j_1 - j_2)!}\\
				\qquad\qquad\qquad\times\sum_{\sigma\in S_N}\Tr\big(U_\sigma^*P_1^{\otimes j_1}\otimes P_2^{\otimes j_2}\otimes Q^{\otimes(N - j_1 - j_2)}U_\sigma\Gamma\big)
			\end{multlined}\\
			&= \dfrac{j_1j_2}{N(N - 1)}\Tr\big(P_1^{\otimes j_1}\otimes_\textmd{s}P_2^{\otimes j_2}\otimes_\textmd{s}Q^{\otimes(N - j_1 - j_2)}\Gamma\big)\\
			&= \dfrac{j_1j_2}{N(N - 1)}.
		\end{align*}
	\end{proof}
	
	\begin{proof}[Proof of Lemma~\ref{lemma:state_one_projection_new_state}]
		Define
		\begin{equation*}
			\Gamma_{j} = {N \choose j}\Tr_{j+1 \rightarrow N}\big(P^{\otimes j}\otimes Q^{\otimes (N - j)}\Gamma\big).
		\end{equation*}
		Firstly, it is easy to verify that $\Gamma_{j}$ is a symmetric state and that it lives in the desired space. Secondly, using again the symmetry of $\Gamma$ and the cyclicity of the trace, we have
		\begin{align*}
			\Gamma_{j}^{(2)} &= {N \choose j}\Tr_{3\rightarrow N}\big(P^{\otimes j}\otimes Q^{\otimes (N - j)}\Gamma\big)\\
			&= {N \choose j}\dfrac{1}{(N - 2)!}\sum_{\sigma\in S_{N - 2}}\Tr_{3\rightarrow N}\big(P^{\otimes 2}\otimes U_\sigma^*P^{\otimes(j - 2)}\otimes Q^{\otimes(N - j)}U_\sigma\Gamma\big)\\
			&= {N \choose j}\dfrac{(j - 2)!(N - j)!}{(N - 2)!}\Tr_{3\rightarrow N}P^2\big(P^{\otimes 2}\otimes \big(P^{\otimes(j - 2)}\otimes_\textmd{s} Q^{\otimes(N - j)}\big)\Gamma\big)P^{\otimes 2}\\
			&= \dfrac{N(N - 1)}{j(j - 1)}P^{\otimes 2}\Tr_{3\rightarrow N}\big(P^{\otimes j}\otimes_\textmd{s}Q^{\otimes(N - j)}\Gamma\big)P^{\otimes 2},
		\end{align*}
		which is the desired claim.
	\end{proof}
	
	\section{Proof of the main technical estimate: Lemma~\ref{lemma:bound_interaction_low_occupancy}}
	
	\label{section:proof_main_theorem_main_lemma}
	
	This section is dedicated to the proof of Lemma~\ref{lemma:bound_interaction_low_occupancy}. Although the proof used in the exponential case covers polynomial scalings as well, we provide first a simpler one in that case for the reader's convenience.
	
	\begin{proof}[Proof of Lemma~\ref{lemma:bound_interaction_low_occupancy} in the polynomial case]
		Thanks to the Cauchy--Schwarz inequality and by definition of $\Gamma_{\underline{J}^{(i_1,i_2;i_1',i_2')}}$, we have
		\begin{align*}
			\big\vert\Tr\big(P_{i_1}\otimes P_{i_2}w_NP_{i_1'}\otimes P_{i_2'}\Gamma_{\underline{J}^{(i_1,i_2;i_1',i_2')}}^{(2)}\big)\big\vert &\leq \dfrac{N^{3\varepsilon}}{2}\Tr\big(P_{i_1}\otimes P_{i_2}\vert w_N\vert P_{i_1}\otimes P_{i_2}\Gamma_{\underline{J}^{(i_1,i_2)}}^{(2)}\big)\\
			&\phantom{\leq} + \dfrac{N^{-3\varepsilon}}{2}\Tr\big(P_{i_1'}\otimes P_{i_2'}\vert w_N\vert P_{i_1'}\otimes P_{i_2'}\Gamma_{\underline{J}^{(i_1',i_2')}}^{(2)}\big).
		\end{align*}
		Then, using Proposition~\ref{prop:plane_wave_estimate} and Lemma~\ref{lemma:state_two_projections_full_trace}, we deduce
		\begin{multline*}
			\big\vert\Tr\big(P_{i_1}\otimes P_{i_2}w_NP_{i_1'}\otimes P_{i_2'}\Gamma_{\underline{J}^{(i_1,i_2;i_1',i_2')}}^{(2)}\big)\big\vert \leq C N^{2\min(i_1,i_2)\varepsilon}N^{3\varepsilon}\dfrac{(j_{i_1} + 1)(j_{i_2} + 1)}{N(N - 1)}\Tr\Gamma_{\underline{J}^{(i_1,i_2)}}\\
			\phantom{\leq} + CN^{2\min(i_1',i_2')\varepsilon}N^{-3\varepsilon}\dfrac{(j_{i_1'} + 1)(j_{i_2'} + 1)}{N(N - 1)}\Tr\Gamma_{\underline{J}^{(i_1',i_2')}}.
		\end{multline*}
		On the one hand, using that $j_{i_1} < N^{1 - \delta\varepsilon}$ and Lemma~\ref{lemma:state_two_projections_full_trace}, we have
		\begin{align*}
			N^{2\min(i_1,i_2)\varepsilon}N^{3\varepsilon}\dfrac{(j_{i_1} + 1)(j_{i_2} + 1)}{N(N - 1)}\Tr\Gamma_{\underline{J}^{(i_1,i_2)}} &\leq N^{(5 - \delta)\varepsilon}N^{2(i_2 - 1)\varepsilon}\dfrac{j_{i_2} + 1}{N}\Tr\Gamma_{\underline{J}^{(i_1,i_2)}}\\
			&\leq CN^{(5 - \delta)\varepsilon}\Tr\big(P_{i_2}(1 + h)P_{i_2}\Gamma_{\underline{J}^{(i_1,i_2)}}^{(1)}\big).
		\end{align*}
		On the other hand,
		\begin{equation*}
			N^{2\min(i_1',i_2')\varepsilon}N^{-3\varepsilon}\dfrac{(j_{i_1'} + 1)(j_{i_2'} + 1)}{N(N - 1)}\Tr\Gamma_{\underline{J}^{(i_1',i_2')}} \leq CN^{-\varepsilon}\Tr\big(P_{i_1'}\otimes P_{i_2'}(1 + h_1)P_{i_1'}\otimes P_{i_2'}\Gamma_{\underline{J}^{(i_1',i_2')}}^{(2)}\big).
		\end{equation*}
		Taking $\delta \geq 6$, and using 
		\begin{equation*}
			\sum_{\underline{J}}P_{i_2}\Gamma_{\underline{J}^{(i_1,i_2)}}^{(1)}P_{i_2} \leq P_{i_2}\Gamma^{(1)}P_{i_2}
		\end{equation*}
		and
		\begin{equation*}
			\sum_{\underline{J}}P_{i_1'}\otimes P_{i_2'}\Gamma_{\underline{J}^{(i_1',i_2')}}^{(2)}P_{i_1'}\otimes P_{i_2'} = P_{i_1'}\otimes P_{i_2'}\Gamma^{(2)}P_{i_1'}\otimes P_{i_2'},
		\end{equation*}
		we therefore get
		\begin{multline*}
			\sum_{\underline{J}}\big\vert\Tr\big(P_{i_1}\otimes P_{i_2}w_NP_{i_1'}\otimes P_{i_2'}\Gamma_{\underline{J}^{(i_1,i_2)},\underline{J}^{i_1',i_2'}}^{(2)}\big)\big\vert \leq CN^{-\varepsilon}\Tr\big(P_{i_2}(1 + h)P_{i_2}\Gamma^{(1)}\big)\\
			+ CN^{-\varepsilon}\Tr\big(P_{i_1'}\otimes P_{i_2'}(1 + h_1)P_{i_1'}\otimes P_{i_2'}\Gamma^{(2)}\big).
		\end{multline*}
		Finally, we sum over $i_1,i_2,i_1',i_2'$, and use that $M$ is a constant to obtain \eqref{eq:bound_interaction_low_occupancy}.
	\end{proof}
	
	In the exponential case, $M$ is no longer bounded, which prevents us from using the previous simple proof (except for $\kappa$ small), and we consequently have to be much more precise.
	
	\begin{proof}[Proof of Lemma~\ref{lemma:bound_interaction_low_occupancy} in the exponential case]
		We distinguish between $\vert i_1 - i_2\vert \leq 2$ and $\vert i_1 - i_2\vert > 2$, and similarly for $i_1',i_2'$ and $i_1,i_1'$. Because we only to use $j_{i_1} \leq N^{1- \delta\varepsilon}$ when $\vert i_1 - i_2\vert \leq 2$, $\vert i_1' - i_2'\vert \leq 2$ and $\vert i_1 - i_1'\vert \leq 2$, we begin by treating the other cases. Namely, we first bound
		\begin{equation}
			\label{eq:interaction_low_occupancy_bound_far_1}
			\sum_{\underline{J}}\sum_{\substack{i_1\leq i_2\\ i_1 + 2 < i_2}}\sum_{\substack{i_1'\leq i_2'}}\big\vert\Tr\big(P_{i_1}\otimes P_{i_2}w_NP_{i_1'}\otimes P_{i_2'}\Gamma_{\underline{J}^{(i_1,i_2;i_1',i_2')}}^{(2)}\big)\big\vert
		\end{equation}
		and
		\begin{equation}
			\label{eq:interaction_low_occupancy_bound_far_2}
			\sum_{\underline{J}}\sum_{\substack{i_1\leq i_2\\ \vert i_1 - i_2\vert \leq 2}}\sum_{\substack{i_1'\leq i_2'\\ \vert i_1' - i_2'\vert \leq 2\\ i_1 + 2 < i_1'}}\big\vert\Tr\big(P_{i_1}\otimes P_{i_2}w_NP_{i_1'}\otimes P_{i_2'}\Gamma_{\underline{J}^{(i_1,i_2;i_1',i_2')}}^{(2)}\big)\big\vert.
		\end{equation}
		In both cases we can take the sum over $i_1$ to start from $1$ rather than $\imax + 1$ for simplicity. Thanks to the Cauchy--Schwarz inequality, Lemma~\ref{lemma:state_two_projections_full_trace} and Proposition~\ref{prop:plane_wave_estimate}, we have
		\begin{multline}
			\label{eq:interaction_low_occupancy_bound_key_cauchy_schwarz}
			\big\vert\Tr\big(P_{i_1}\otimes P_{i_2}w_NP_{i_1'}\otimes P_{i_2'}\Gamma_{\underline{J}^{(i_1,i_2;i_1',i_2')}}^{(2)}\big)\big\vert \leq \tau CN^{(i_1 + i_1')\varepsilon}\dfrac{(j_{i_1} + 1)(j_{i_2} + 1)}{N(N - 1)}\Tr\Gamma_{\underline{J}^{(i_1,i_2)}}\\
			+ \tau^{-1}CN^{(i_1 + i_1')\varepsilon}\dfrac{(j_{i_1'} + 1)(j_{i_2'} + 1)}{N(N - 1)}\Tr\Gamma_{\underline{J}^{(i_1',i_2')}},
		\end{multline}
		for all $\tau > 0$. Thus, an appropriate choice of $\tau$ and an application of Lemma~\ref{lemma:state_two_projections_full_trace} allow us to write
		\begin{multline*}
			\big\vert\Tr\big(P_{i_1}\otimes P_{i_2}w_NP_{i_1'}\otimes P_{i_2'}\Gamma_{\underline{J}^{(i_1,i_2;i_1',i_2')}}^{(2)}\big)\big\vert\\
			\begin{aligned}[t]
				&\leq CN^{(i_1 - i_2 + 1)\varepsilon}N^{(i_1' - i_2' + 1)\varepsilon}\Tr\big(P_{i_2}(1 + h)P_{i_2}\otimes\big(P_{i_1'} + N^{-1}\big)\Gamma_{\underline{J}^{(i_1,i_2)}}^{(2)}\big)\\
				&\phantom{\leq} + CN^{(i_1 - i_2 + 1)\varepsilon}N^{(i_1' - i_2' + 1)\varepsilon}\Tr\big(P_{i_2'}(1 + h)P_{i_2'}\otimes\big(P_{i_1} + N^{-1}\big)\Gamma_{\underline{J}^{(i_1',i_2')}}^{(2)}\big).
			\end{aligned}
		\end{multline*}
		The reason we have the factor $N^{-1}$ in the right-hand side is to deal with the possibility of $j_{i_1}$ and $j_{i_1'}$ being null. Using that
		\begin{equation}
			\label{eq:interaction_low_occupancy_sum_over_J}
			\sum_{\underline{J}}P_{i_1}\otimes P_{i_2}\Gamma_{\underline{J}^{(i_1,i_2)}}^{(2)}P_{i_1}\otimes P_{i_2} = P_{i_1}\otimes P_{i_2}\Gamma^{(2)}P_{i_1}\otimes P_{i_2},
		\end{equation}
		we may carry out the $\underline{J}$ sum in \eqref{eq:interaction_low_occupancy_bound_far_1} first to obtain
		\begin{multline*}
			\sum_{\underline{J}}\sum_{\substack{i_1 + 2 < i_2}}\sum_{\substack{i_1'\leq i_2'}}\big\vert\Tr\big(P_{i_1}\otimes P_{i_2}w_NP_{i_1'}\otimes P_{i_2'}\Gamma_{\underline{J}^{(i_1,i_2;i_1',i_2')}}^{(2)}\big)\big\vert\\
			\begin{aligned}[t]
				&\leq \sum_{i_1 + 2 < i_2}\sum_{i_1'\leq i_2'}CN^{(i_1 - i_2 + 1)\varepsilon}N^{(i_1' - i_2' + 1)\varepsilon}\Tr\big(P_{i_2}hP_{i_2}\otimes\big(P_{i_1'} + N^{-1}\big)\Gamma^{(2)}\big)\\
				&\phantom{\leq} + \sum_{i_1 + 2 < i_2}\sum_{i_1'\leq i_2'}CN^{(i_1 - i_2 + 1)\varepsilon}N^{(i_1' - i_2' + 1)\varepsilon}\Tr\big(P_{i_2'}hP_{i_2'}\otimes\big(P_{i_1} + N^{-1}\big)\Gamma^{(2)}\big).
			\end{aligned}
		\end{multline*}
		Then, carrying out the remaining four sums using the geometric series formula, the resolution of the identity
		\begin{equation*}
			\sum_{i = 1}^{M + 1}P_i = 1
		\end{equation*}
		and the fact that $M$ satisfies \eqref{eq:relation_M_epsilon_expo},we obtain
		\begin{equation}
			\label{eq:interaction_low_occupancy_bound_far_1_estimate}
			\sum_{\underline{J}}\sum_{\substack{i_1 + 2 < i_2}}\sum_{\substack{i_1'\leq i_2'}}\big\vert\Tr\big(P_{i_1}\otimes P_{i_2}w_NP_{i_1'}\otimes P_{i_2'}\Gamma_{\underline{J}^{(i_1,i_2;i_1',i_2')}}^{(2)}\big)\big\vert \leq CN^{-\varepsilon}\Tr\big((1 + h)\Gamma^{(1)}\big).
		\end{equation}
		Hence, we have bounded \eqref{eq:interaction_low_occupancy_bound_far_1} as desired.
		
		To bound \eqref{eq:interaction_low_occupancy_bound_far_2}, we again use \eqref{eq:interaction_low_occupancy_bound_key_cauchy_schwarz} to write
		\begin{multline*}
			\big\vert\Tr\big(P_{i_1}\otimes P_{i_2}w_NP_{i_1'}\otimes P_{i_2'}\Gamma_{\underline{J}^{(i_1,i_2;i_1',i_2')}}^{(2)}\big)\big\vert\\
			\begin{aligned}[t]
				&\leq CN^{(i_1 - i_1' + 1)\varepsilon}N^{(i_1' - i_2' + 1)\varepsilon}\Tr\big(P_{i_1'}(1 + h)P_{i_1'}\otimes P_{i_1}\Gamma_{\underline{J}^{(i_1,i_2)}}^{(2)}\big)\\
				&\phantom{\leq} + CN^{(i_1 - i_1' + 1)\varepsilon}N^{(i_1' - i_2' + 1)\varepsilon}\Tr\big(P_{i_2'}(1 + h)P_{i_2'}\otimes\big(P_{i_2} + N^{-1}\big)\Gamma_{\underline{J}^{(i_1',i_2')}}^{(2)}\big),
			\end{aligned}
		\end{multline*}
		when $j_{i_1'} \geq 1$. When $j_{i_1'} = 0$, the previous expression cannot by made true by simply replacing $P_{i_1'}$ by $1/\sqrt{N}$, and we instead write
		\begin{multline*}
			\big\vert\Tr\big(P_{i_1}\otimes P_{i_2}w_NP_{i_1'}\otimes P_{i_2'}\Gamma_{\underline{J}^{(i_1,i_2;i_1',i_2')}}^{(2)}\big)\big\vert\\
			\begin{aligned}[t]
				&\leq CN^{2\varepsilon}N^{-(1 + \kappa)/2}\Tr\big(P_{i_1}(1 + h)P_{i_1}\Gamma_{\underline{J}^{(i_1,i_2)}}^{(1)}\big)\\
				&\phantom{\leq} + CN^{2\varepsilon}N^{-(1 - \kappa)/2}\Tr\big(\big(P_{i_2} + N^{-1}\big)\otimes P_{i_2'}(1 + h)P_{i_2'}\Gamma_{\underline{J}^{(i_1',i_2')}}^{(2)}\big).
			\end{aligned}
		\end{multline*}
		Carrying out the sums similarly as before yields
		\begin{multline}
			\label{eq:interaction_low_occupancy_bound_far_2_estimate}
			\sum_{\underline{J}}\sum_{\substack{i_1\leq i_2\\ \vert i_1 - i_2\vert \leq 2}}\sum_{\substack{i_1'\leq i_2'\\ \vert i_1' - i_2'\vert \leq 2\\ i_1 + 2 < i_1'}}\big\vert\Tr\big(P_{i_1}\otimes P_{i_2}w_NP_{i_1'}\otimes P_{i_2'}\Gamma_{\underline{J}^{(i_1,i_2;i_1',i_2')}}^{(2)}\big)\big\vert\\
			\leq C(N^{-\varepsilon} + N^{2\varepsilon -(1 - \kappa)/2})\Tr\big((1 + h)\Gamma^{(1)}\big).
		\end{multline}
		Thus, we have shown that \eqref{eq:interaction_low_occupancy_bound_far_2} is bounded as claimed in Lemma~\ref{lemma:bound_interaction_low_occupancy}.
		
		What now remains is to deal with
		\begin{equation*}
			\sum_{\underline{J}}\sum_{\substack{i_1\leq i_2\\ \vert i_1 - i_2\vert \leq 2}}\sum_{\substack{i_1'\leq i_2'\\ \vert i_1' - i_2'\vert \leq 2\\ \vert i_1 - i_1'\vert \leq 2}}\big\vert\Tr\big(P_{i_1}\otimes P_{i_2}w_NP_{i_1'}\otimes P_{i_2'}\Gamma_{\underline{J}^{(i_1,i_2;i_1',i_2')}}^{(2)}\big)\big\vert,
		\end{equation*}
		where the sum over $i_1$ starts with $i_1 = \imax + 1$, meaning that we always have $j_{i_1} \leq N^{1 - \delta\varepsilon}$. Using once more \eqref{eq:interaction_low_occupancy_bound_key_cauchy_schwarz}, as well as Lemma~\ref{lemma:state_two_projections_full_trace} and the fact that $j_{i_1} \leq N^{1 - \delta\varepsilon}$, we obtain
		\begin{multline*}
			\big\vert\Tr\big(P_{i_1}\otimes P_{i_2}w_NP_{i_1'}\otimes P_{i_2'}\Gamma_{\underline{J}^{(i_1,i_2;i_1',i_2')}}^{(2)}\big)\big\vert\\
			\begin{aligned}[t]
				&\leq CN^{(i_1 - i_2 + 1)\varepsilon}N^{(i_1' - i_2' + 1)\varepsilon}N^{-\delta\varepsilon/2}\Tr\big(P_{i_2}hP_{i_2}\Gamma_{\underline{J}^{(i_1,i_2)}}^{(1)}\big)\\
				&\phantom{\leq} + CN^{(i_1 - i_2 + 1)\varepsilon}N^{(i_1' - i_2' + 1)\varepsilon}N^{-\delta\varepsilon/2}\Tr\big(P_{i_2'}hP_{i_2'}\Gamma_{\underline{J}^{(i_1',i_2')}}^{(1)}\big).
			\end{aligned}
		\end{multline*}
		Carrying out the sums as above, we finally obtain
		\begin{equation}
			\label{eq:interaction_low_occupancy_bound_close_estimate}
			\sum_{\underline{J}}\sum_{\substack{i_1\leq i_2\\ \vert i_1 - i_2\vert \leq 2}}\sum_{\substack{i_1'\leq i_2'\\ \vert i_1' - i_2'\vert \leq 2\\ \vert i_1 - i_1'\vert \leq 2}}\big\vert\Tr\big(P_{i_1}\otimes P_{i_2}w_NP_{i_1'}\otimes P_{i_2'}\Gamma_{\underline{J}^{(i_1,i_2;i_1',i_2')}}^{(2)}\big)\big\vert \leq CN^{2\varepsilon - \delta\varepsilon/2}\Tr\big((1 + h)\Gamma^{(1)}\big).
		\end{equation}
		Gathering the estimates \eqref{eq:interaction_low_occupancy_bound_far_1_estimate}--\eqref{eq:interaction_low_occupancy_bound_close_estimate}, and taking $\delta \geq 6$ yields \eqref{eq:bound_interaction_low_occupancy}.
	\end{proof}
	
	\printbibliography
	
	\end{document}